\documentclass{article}
\usepackage{color,ulem,soul}
\usepackage{amsmath, amssymb, amsthm, braket, latexsym, mathtools}
\usepackage{subfiles,subcaption}
\usepackage{bm,array,arydshln}
\usepackage{graphicx}

\usepackage{comment}

\newtheorem{theorem}{Theorem}
\newtheorem{lemma}{Lemma} 
\newtheorem{prop}{Proposition}

\newcommand{\HA}{\mathcal{H}_{A}}
\newcommand{\HV}{\mathcal{H}_{V}}
\newcommand{\HVI}{\mathcal{H}_{V}^{\mathrm{inv}}}
\newcommand{\HAI}{\mathcal{H}_{A}^{\mathrm{inv}}}
\newcommand{\Deg}{\mathrm{deg}}

\DeclareRobustCommand{\Erase}{\bgroup\markoverwith{\textcolor{red}{\rule[.5ex]{2pt}{0.4pt}}}\ULon}

\begin{document}
\title{{\bf Pulsation of quantum walk on Johnson graph}
\vspace{0mm}} 

\author{
  Taisuke Hosaka$^{*}$, Etsuo Segawa$^{\dagger}$
  \\
  \small $^{*,\dagger}$Graduate School of Environment and Information Sciences \\
  \small Yokohama National University \\
  \small Hodogaya, Yokohama, 240-8501, Japan \\
}

\date{\empty }

\maketitle
\thispagestyle{empty}




\par\noindent


\begin{abstract}
  We propose a phenomenon of discrete-time quantum walks on graphs called the pulsation, which is a generalization of a phenomenon in the quantum searches.
  This phenomenon is discussed on a composite graph formed by two connected graphs $G_{1}$ and $G_{2}$.
  The pulsation means that the state periodically transfers between $G_{1}$ and $G_{2}$ 
  with the initial state of the uniform superposition on $G_1$.
  In this paper, we focus on the case for the Grover walk where $G_{1}$ is the Johnson graph and $G_{2}$ is a star graph.
  Also, the composite graph is constructed by identifying an arbitrary vertex of the Johnson graph with the internal vertex of the star graph.
  In that case,
  we find  the pulsation with  $O(\sqrt{N^{1+1/k}})$ periodicity, 
  where $N$ is the number of vertices of the Johnson graph.
The proof is based on Kato's perturbation theory in finite-dimensional vector spaces.

  \mbox{} 
\end{abstract}


\section{Introduction}
\label{sec_intro}
Quantum walks (QWs) are considered as the quantum counterpart of the classical random walks (RWs).
QWs have an unique features that are not presented in RWs such as localization \cite{IK05, Ki22, M15}, periodicity \cite{HKSS17,KSTY18} and ballistic spreading \cite{K02,K05,KLS13}.
Because of its characteristics, QWs is widely used for various field.
One of the tasks to study QWs is the state transfer between two vertices of a graph \cite{CGGV15,G12, GGKL20}.
The aim is to transfer the state to a given vertex with high probability after certain steps.
The phenomenon has important application for quantum communication and search algorithms.
Also, this problem has been investigated about discrete-time QWs \cite{KS22,SS23}.
Quantum search algorithms are one of the application of QWs.
Quantum search algorithms have been studied that QWs provide quadratic faster than corresponding classical search algorithms via RWs on two-dimentional torus \cite{AKR05}, hypercube \cite{SKW03}, the Johnson graph \cite{TSP22D} and other several graphs \cite{ADZ93,AAKV01,A07,AGFK20,S04}.
This paper represents that we propose a new property of quantum walks called ``pulsation'' on graphs.
Let us explain this phenomenon precisely as follows:
For two connected graphs $G_{1}, G_{2}$, we consider a composite graph formed by two connected graph $G_{1}$ and $G_{2}$.
In addition, we assume that all initial states are on $G_{1}$.
When the states are evolved after taking a certain steps, almost all states transfer from $G_{1}$ to $G_{2}$, that is, the asymptotic probability of finding $G_{2}$ is $1$.
After another certain steps, the states return to $G_{1}$.
Then the pulsation is the property that the above phenomenon, which a QWer goes back and forth between $G_{1}$ and $G_{2}$, repeatedly occurs.
In other word, the asymptotic finding probabilities of $G_{1}$ and $G_{2}$ switches  $0$ to $1$ with a certain periodicity.
This phenomenon can be considered as a kind of generalization of spatial quantum search algorithm \cite{CG04, R18}, because the target vertex can be regarded as the isolated vertex $G_2$ embedded in $G_1$.
The goal of the spatial search algorithm is to collect the almost all states to the marked vertices quickly from the uniform initial state on a graph.
Then we also set that the initial state is uniform on $G_{1}$, and assume that $G_{2}$ is smaller than $G_{1}$ so that $G_{2}$ can be regarded as the {\it target} like marked vertices.
In this settings, from the view point of quantum search algorithms, the aim is to transfer from $G_{1}$ to a {\it geometric perturbation} $G_{2}$ quickly by using the pulsation which is a property of quantum walks.

The previous work \cite{HPS24} showed the existence of the pulsation on the case where $G_{1}$ is the complete graph with $n$ vertices, denoted by $K_{n}$, and $G_{2}$ is the star graph with $m$ leaves, denoted by $S_{m}$.
To reveal the mathematical structure of the pulsation in more general situation, fixing $S_{m}$ as $G_{2}$, we treat the Johnson graph $J(n,k) \; (n \in \mathbb{N}, \; 0 \leq k \leq n-1)$ as $G_{1}$, which is 
a family of an infinite series of graphs including the complete graphs. 
We rigorously obtain that for a fixed $k$ and $n \gg k$, the pulsation occurs, that is, the states back and forth between $J(n,k)$ and $S_{m}$ with the asymptotic probability 1 by taking $O(\sqrt{N^{1+1/k}})$ steps, where $N$ is the number of vertices of the Johnson graph.
In other words, $S_{m}$ absorbs almost all state of QW from the uniform state on $J(n,k)$.
In particular, this result tells us that almost all the energy of QW can be stored in only one edge for the setting of $m=1$ at a certain time.

The paper \cite{TSP22D} deals with detection of the target vertex in the Johnson graph. On the other hand, our paper gives a geometric perturbation (i.e. $S_{m}$) to the Johnson graph from the outside and regard that as the target.
These difference affect the asymptotic probability and the optimal time: 
the paper \cite{TSP22D} showed that the asymptotic probability is $1/2$ by taking $O(\sqrt{N})$ steps, while 
our results imply that the target $S_m$ can be found with the high probability which is almost $1$ in the property of the pulsation, instead, the speed is slightly slow down, that is, $O(\sqrt{N^{1+1/k}})$.

The rest of this paper is organized as follows:
We set the notation and our quantum walk model in Section \ref{sec_setting}.
Section \ref{sec_result} gives the main results.
In Section \ref{sec_proof}, we shows the proofs of the main results.
Section \ref{sec: summary} presents the summary and discussion.

\section{Setting of the model}
\label{sec_setting}
Let $G=(V,E)$ be a simple connected graph, where $V=V(G)$ is the set of vertices and $E=E(G)$ is the set of edges.
Let $A=A(G)$ be a set of symmetric arcs induced by $E(G)$.
The origin and terminus vertices of $a \in A$ are denoted by $o(a)$ and $t(a)$, respectively.
We write the inverse of arc of $a \in A$ as $a^{-1}$.
Remark that $t(a)=o(a^{-1}), t(a^{-1})=o(a)$.
The degree of $v \in V$ is denoted by $\Deg(v)=|\left\{a \in A\, \middle| \, t(a)=v  \right\}|$.
Let $J(n,k)$ be the Johnson graph.
The set of the vertices of the Johnson graph, $V(J(n,k))$, is the set of $k$-element subsets of $[n]=\{1,2,...,n\}$.
Two vertices $v,v' \in V(J(n,k))$ are adjacent if and only if $|v\cap v'|=k-1$ holds.
Note that the number of vertices $N$ is $\binom{n}{k}$, and $J(n,k)$ is $d$-regular, where 
\begin{align*}
  d=k(n-k).
\end{align*}
The diameter of $J(n,k)$ is $\mathrm{min}\{k,n-k\}$, however, in this manuscript, we assume that $k$ is fixed and $n$ is sufficiently large.
Thus, we consider $J(n,k)$ has the diameter $k$.
Let $S_{m}$ be the star graph, which one center vertex is connected to all the other $m$ vertices and no other pairs of vertices are connected.
Here, let $J(n,k) \wedge_{v^{*}} S_{m}$ be the graph identifying a fixed vertex of $J(n,k)$ with the center vertex of $S_{m}$.
This identified vertex is denoted by $v^{*}$.
Then, we decompose the set of arcs and vertices as follows:
\begin{align*}
  &A:=A(J(n,k) \wedge_{v^*} S_{m})= A_{J} \sqcup A_{S}, \\
  &V:=V(J(n,k) \wedge_{v^*} S_{m})= V_{J} \sqcup V_{S} \sqcup \{v^{*}\},
\end{align*}
where,
\begin{align*}
  A_{J}=A(J(n,k)),&\quad A_{S}=A(S_{m}), \\
  V_{J}=V(J(n,k)) \setminus \{v^{*}\}, &\quad V_{S}=V(S_{m}) \setminus \{v^{*}\}.
\end{align*}

The Hilbert space $\HA$ and $\HV$ is spanned by the set of arcs and vertices,
that is, 
\begin{align*}
  \HA=\mathrm{Span}\{\ket{a}\,|\, a \in A\}, \quad \HV=\mathrm{Span}\{\ket{v}\,|\, v \in V\},
\end{align*}
respectively. 
The time evolution matrix $U$ on $\HA$ is defined by
\begin{align*}
  U=S(2K^{*}K-I_{A}),
\end{align*}
where $S$ is the shift matrix, $K$ is the boundary matrix and $I_{A}$ is the identity matrix on $\HA$.
The shift matrix $S : \HA \rightarrow \HA$ and the boundary matrix $K : \HA \rightarrow \HV$ is defined as follows:
\begin{align*}
  S\ket{a}=\ket{a^{-1}}.
\end{align*}
\begin{equation*}
  K\ket{a}=
  \begin{dcases*}
    \sum_{t(a)=v}\frac{1}{\sqrt{\mathrm{deg}(v)}} \ket{v} &: $v \notin V_{S}$, \\
    0 &: $\text{otherwise}$.
  \end{dcases*}
\end{equation*}
By easily calculation, we can check that $S^{2}=I_{A}$ and $KK^{*}=I_{V}$ hold, where $I_{V}$ is the identity matrix on $\HV$.
From the definition of $S$ and $K$, the time evolution matrix $U$ is described by
\begin{equation*}
  U\ket{a}=
  \begin{dcases*}
    \sum_{o(b)=t(a)} \left(\frac{2}{\mathrm{deg}(t(a))}-\delta_{a^{-1},b}\right) \ket{b} &: $t(b) \notin V_{S}$, \\
    -\ket{b} &: $t(b) \in V_{S}$,
  \end{dcases*}
\end{equation*}
where $\delta_{a,b}$ is the Kronecker delta.
This definition shows that each leaf has its phase reversed after the action of $U$.

The initial state $\ket{\psi_{0}} \in \HA$ is set as the uniform superposition on $J(n,k)$, that is,
\begin{equation}
  \label{eq: initial state}
  \ket{\psi_{0}}=\frac{1}{\sqrt{|A_{J}|}} \sum_{a \in A_{J}} \ket{a}.
\end{equation}
Let $\ket{\psi_{t}}$ be the $t$-th iteration of $U$, that is,
\begin{align*}
  \ket{\psi_{t+1}}=U\ket{\psi_{t}}.
\end{align*}
The finding probability of the star graph after $t$ steps is defined by
\begin{align*}
  p_{s}(t)=\sum _{a \in A_{S}} |\braket{a|\psi_{t}}|^{2}.
\end{align*}
We are particularly interested in finding the optimal running time $\tau$, at which the probability
$p_{s}(\tau)$ of finding star graph $S_{m}$ is maximized.

\section{Main result}
\label{sec_result}
This section shows the probability of finding the star graph and the optimal time $\tau$.
\begin{theorem}
  \label{thm: prob}
  For sufficiently large $N$, there exist a time $\tau$ such that
  \begin{align*}
    &p_{s}(2\ell \cdot \tau)=o(1), \\
    &p_{s}((2\ell+1) \cdot \tau)=1+o(1) \\
  \end{align*}
  for $\ell=0,1,2,...$.
\end{theorem}
\begin{theorem}
  \label{thm: time}
  For sufficiently large $N$, the optimal time $\tau$ is obtained as
  \begin{align*}
    \tau\sim\frac{\pi\sqrt{k\cdot (k!)^{1/k}}}{2\sqrt{2m}}\sqrt{N^{1+1/k}}=O\left(\sqrt{N^{1+1/k}}\right).
  \end{align*}
\end{theorem}
Theorem \ref{thm: prob} shows the existence of the pulsation on this setting, that is, almost all states periodically transfer between the Johnson graph and star graph.
Theorem \ref{thm: time} shows that the complete graph (i.e. $k=1$ case) gives no-speed up while the Johnson graph $J(n,k) \, (k \geq 2)$ gives a speed up between super-linear and sub-quadratic than classical random walk.
Note that $k$ is the diameter of the Johnson graph.
Figure 1 displays the numerical simulation of the probability finding the star graph.

\begin{figure}[htbp]
  \begin{minipage}[b]{0.5\linewidth}
    \centering
    \includegraphics[keepaspectratio, scale=0.55]{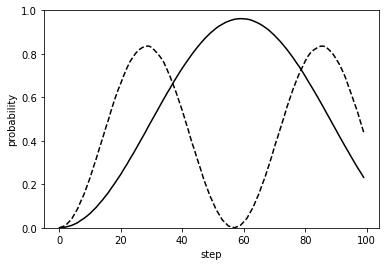}
    \subcaption{}
  \end{minipage}
  \begin{minipage}[b]{0.5\linewidth}
    \centering
    \includegraphics[keepaspectratio, scale=0.55]{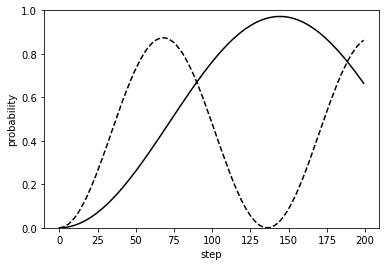}
    \subcaption{}
  \end{minipage} 
  \caption{
    The probability of the finding star graph on the Johnson graph. The solid and dot curves correspond to the case of $m=1$ and $m=5$, respectively.
    (a) and (b) are simulated on $J(15,2)$ and (b) $J(15,3)$, respectively.
  }
\end{figure}

\section{Proof of main results}
\label{sec_proof}

{\bf Proof.}
First, $V_{J}$ is devided into the disjoint sets as follows:
\begin{align*}
  X_{j}=\{v \in V_{J} \,|\, |v \cap v^{*}| =k-j \} \quad (0 \leq j \leq k).
\end{align*}
Note that $j$ is the length of the shortest path between $v^{*}$ and $v \in X_{j}$.
We define
\begin{align*}
    \ket{X_{-1}}=\frac{1}{\sqrt{|X_{-1}|}}\sum_{v \in V_{S}} \ket{v},
  \quad \ket{X_{j}}=\frac{1}{\sqrt{|X_{j}|}}\sum_{v \in X_{j}} \ket{v} \quad (0 \leq j \leq k).
\end{align*}
Remark that
\begin{align*}
    |X_{-1}|=m, \quad 
    |X_{j}|=\dbinom{k}{j}\dbinom{n-k}{j} \quad (0 \leq j \leq k).
\end{align*}
Then, the subspace $\HVI=\mathrm{Span}\{ \ket{X_{j}}\,|\, j=-1,0,1,...,k\}$ is invariant under the adjacency matrix $M$ of $J(n,k) \wedge_{v^{*}} S_{m}$.
By using fact, we have
\begin{align*}
  M\ket{X_{-1}}&=m \ket{X_{0}}, \\
  M\ket{X_{0}}&=c_{1}\ket{X_{1}}+ a_{0}\ket{X_{0}} + m \ket{X_{-1}},\\
  M\ket{X_{j}}&=c_{j+1}\ket{X_{j+1}}+a_{j}\ket{X_{j}}+b_{j-1}\ket{X_{j-1}} \quad (1 \leq j \leq k), 
\end{align*}
where
\begin{align*}
  a_{j}=j(n-2j), \quad b_{j}=(k-j)(n-k-j), \quad c_{j}=j^{2} \quad (0 \leq j \leq k), 
\end{align*}
and we set $\ket{X_{k+1}}=0$ for sake of simplicity.
Remark that $b_{0}=d$ and
\begin{align}
  a_{j}+b_{j}+c_{j}=d \quad &(0 \leq j \leq k), \notag \\
  \label{eq: dist}
  b_{j}|X_{j}|=c_{j+1}|X_{j+1}| \quad &(0 \leq j \leq k).
\end{align}
Let $T$ be the $(k+1) \times (k+1)$ matrix which the adjacency matrix $M$ divided by each degree of $v \in X_{j}$.
Note that $T$ is the transition matrix of random walk on the projected path with self loop, that is, $T$ acts on $\HVI$. 
The matrix form of $T$ is as follows:
\begin{align*}
  T=
  \begin{pmatrix}
    0 & x & 0 & 0 & \dots & 0 \\ 
    p_{1} & r_{1} & q_{1} & 0 & \dots & 0 \\
    0 & p_{2} & r_{2} & q_{2} & \dots & 0 \\
    \vdots & \ddots & \ddots & \ddots & \ddots & \vdots \\
    \vdots & \ddots & \ddots & p_{k-1} & r_{k-1} & q_{k-1} \\
    0 & \dots & \dots & \dots & p_{k} & r_{k} \\
  \end{pmatrix},
\end{align*}
where
\begin{align*}
  p_{j}=\frac{c_{j}}{d}, \quad q_{j}=\frac{b_{j}}{d}, \quad r_{j}=\frac{a_{j}}{d}, \quad x=\frac{d}{d+m} \quad (1\leq j \leq k).
\end{align*}
Note that $p_{j},q_{j}$ and $r_{j}$ is the transfer probability from $X_{j}$ to $X_{j-1},X_{j+1}$ and $X_{j}$, respectively.

Let $\ket{\pi}=\sum_{j=-1}^{k} \pi_{j}\ket{X_{j}}$ be a unit vector such that
\begin{align*}
  &\pi_{-1}=\braket{X_{-1}|\pi}=C_{\pi}, \\
  &\pi_{j}=\braket{X_{j}|\pi}=\frac{d}{m}|X_{j}|C_{\pi},
\end{align*}
where $C_{\pi}=m/(m+|A_{J}|)$ is the normalized constant.
Then, by using Eq.\ (\ref{eq: dist}), we see that $\ket{\pi}$ satisfies the following equation
\begin{align}
  \label{eq: DBC}
  q_{j}\pi_{j}=p_{j}\pi_{j+1},
\end{align}
which is called the detail balanced condition.
We put a matrix $J$ as follows:
\begin{align*}
  J=D_{\pi}^{1/2}TD_{\pi}^{-1/2},
\end{align*}
where $D_{\pi}^{1/2}\ket{X_{j}}=\sqrt{\pi_{j}}\ket{X_{j}} \quad (j=-1,0,...,k)$ holds.
Then we see the following propositions.
\begin{prop}
  \label{prop: isospec}
  $T$ and $J$ are isospectral, that is,
  \begin{align*}
    \mathrm{Spec}(T)=\mathrm{Spec}(J).
  \end{align*}
  Let us set $\mu \in \mathrm{Spec}(T)$. Then, it follows that
  \begin{align*}
      \mathrm{Ker}(\mu-T)=\mathrm{Ker}(\mu-D_{\pi}^{1/2}).
  \end{align*}
\end{prop}
\begin{prop}
  \label{prop: SMT}
  {\rm (\cite{HKSS14})} \\
  Let $\cos \theta \in \mathrm{Spec}(J)$ and $\ket{g} \in \mathrm{Ker}(\cos \theta-J)$ with $\cos \theta \in (-1,1)$.
  Then, $e^{\pm i \theta} \in \mathrm{Spec}(U)$,
  \begin{align*}
    \ket{\psi_{\pm \theta}}
    =\frac{1}{\sqrt{2}|\sin \theta|}\left(K^{*}-e^{\pm i \theta}SK^{*}\right)\ket{g} \in \mathrm{Ker}(e^{\pm i \theta}-U),
  \end{align*}
  and $||\psi_{\pm \theta}||=||g||$.
\end{prop}

Let us set 
\begin{align*}
\mathcal{A}_j &=\{ a\in A \;|\; t(a)\in X_j,\;o(a)\in X_j \}, \\
\mathcal{B}_j &=\{ a\in A \;|\; t(a)\in X_{j+1},\;o(a)\in X_j \}, \\
\mathcal{C}_j &=\{ a\in A \;|\; t(a)\in X_{j-1},\;o(a)\in X_j \},
\end{align*}
for $j=0,1,\dots,k$ and 
\[ S_+=\{ a\in A \;|\; t(a)\in X_0,\;o(a)\in X_{-1}\},\;S_-=\{ a\in A \;|\; t(a)\in X_{-1},\;o(a)\in X_{0}\}. \]
Remark that $\mathcal{A}_j$, $\mathcal{B}_j$ and $\mathcal{C}_j$ are regarded as the self loop, the arc to $j+1$ and the arc to $j-1$ from the vertex $j$ in the path with length $k$ whose boundaries are labeled by $0$ and $k$; $S_\pm$ are regarded as the additional incoming and outgoing arcs connected to the boundary $0$ of that path. Thus we set $t(\mathcal{B}_j)=o(\mathcal{C}_{j+1})=j+1$,  $t(\mathcal{A}_j)=o(\mathcal{A}_j)=j$, and $t(S_+)=o(S_-)=0$, $o(S_+)=t(S_-)=-1$.

For $j=0,1,...,k$, we introduce the following normal vectors induced by the above subsets 
\[
A^{inv}=\{ \mathcal{A}_j,\mathcal{B}_j,\mathcal{C}_j,S_\pm,\;|\;j=0,1,\dots,k \}
\] 
as follows:
\begin{align*}
  \ket{\mathcal{A}_{j}}=\frac{1}{\sqrt{a_{j}|X_{j}|}}\sum_{a\in \mathcal{A}_j} \ket{a}, \quad
  \ket{\mathcal{B}_{j}}=\frac{1}{\sqrt{b_{j}|X_{j}|}}\sum_{a\in \mathcal{B_j}} \ket{a}, \quad
  \ket{\mathcal{C}_{j}}=\frac{1}{\sqrt{c_{j}|X_{j}|}}\sum_{a\in \mathcal{C}_j} \ket{a}.
\end{align*}
\begin{align*}
  \ket{S_{+}}=\frac{1}{\sqrt{m}}\sum_{a\in S_+} \ket{a}, \quad
  \ket{S_{-}}=\frac{1}{\sqrt{m}}\sum_{a\in S_-} \ket{a},
\end{align*}
where we set $\ket{\mathcal{A}_{0}}=\ket{\mathcal{B}_{k}}=\ket{\mathcal{C}_{0}}=0$.
Then, it is immediately to see that they are orthogonal to each other, and $\HAI=\mathrm{Span}\{\ket{\mathcal{A}_{j}}, \ket{\mathcal{B}_{j}}, \ket{\mathcal{C}_{j}}, \ket{S_{\pm}} \,|\, j=0,1,...,k\}$ is invariant under $U$.

Here, for simplicity, we notate
\begin{align*}
  \braket{X_{j}|f}=:f(j),
\end{align*}
for any $f \in \HVI$ and $X_j\in \{X_s;\;|\;s=-1,0,\dots,k\}$.
Likewise, we notate
\begin{align*}
  \braket{\mathcal{W}|\psi}=:\psi(\mathcal{W}) 
\end{align*}
for any $\psi\in \mathcal{H}^{inv}_A$ and $\mathcal{W}\in A^{inv}$.
Let $p(\mathcal{W})$ be the transition probability from $X_{o(\mathcal{W})}$ to $X_{t(\mathcal{W})}$ on $\HVI$, that is,
\begin{equation*}
  p(\mathcal{W})=
  \begin{dcases*}
    p_{j} &: $\mathcal{W}= \mathcal{C}_{j}$ \\
    q_{j} &: $\mathcal{W}= \mathcal{B}_{j}$ \\
    r_{j} &: $\mathcal{W}= \mathcal{A}_{j}$ \\
    x &: $\mathcal{W}= \mathcal{B}_{0}$ \\
  \end{dcases*}
\end{equation*}
Then, we define $M(\mathcal{W})$ as follows
\begin{align*}
  M(\mathcal{W})=\pi_{o(\mathcal{W})}p(\mathcal{W})=\pi_{t(\mathcal{W}^{-1})}p(\mathcal{W}^{-1})
\end{align*}
Note that the second equal sign of above equation is from Eq.\ (\ref{eq: DBC}).
By using $M(\mathcal{W})$, we get the following lemma.
\begin{lemma}
  \label{lem: element}
  Let $\cos \theta \in \mathrm{Spec}(T)$ and $\ket{f} \in \mathrm{Ker}(\cos \theta -T)$ with $\cos \theta \in (-1,1)$.
  Then, $e^{\pm i \theta} \in \mathrm{Spec}(U)$,
  \begin{align*}
    \psi_{\pm \theta}(\mathcal{W})=\frac{\sqrt{M(\mathcal{W})}}{\sqrt{2}|\sin \theta|}(f(t(\mathcal{W}))-e^{\pm i \theta}f(o(\mathcal{W}))) \in \mathrm{Ker}(e^{\pm i \theta}-U)
  \end{align*}
\end{lemma}
for $\mathcal{W} \in A^{\mathrm{inv}}$.
\begin{proof}
  Combining Proposition \ref{prop: isospec} with \ref{prop: SMT},
  we get $e^{\pm i \theta} \in \mathrm{Spec}(U)$ immediately.
   Let $\cos \theta \in \mathrm{Spec}(T)$ and $\ket{f} \in \mathrm{Ker}(\cos \theta -T)$ with $\cos \theta \in (-1,1)$.
  From the Proposition \ref{prop: SMT}, it follows that
  \begin{align}
    \label{eq: elem_W}
    \psi_{\pm \theta}(\mathcal{W})
    &=\frac{1}{\sqrt{2}|\sin \theta|}\left(\bra{\mathcal{W}}K^{*}\ket{g}-e^{\pm i \theta}\bra{\mathcal{W}}SK^{*}\ket{g}\right) 
  \end{align}
  for $\mathcal{W} \in A^{\mathrm{inv}}$.
  Focusing on the first term of Eq.\ (\ref{eq: elem_W}), we have
  \begin{align}
    \label{eq: first_term}
    \bra{\mathcal{W}}K^{*}\ket{g}
    &=\frac{1}{\sqrt{|\mathcal{W}|}}\sum_{a \in \mathcal{W}}\bra{a}K^{*}\ket{g} \notag \\
    &=\frac{1}{\sqrt{|\mathcal{W}|}}\sum_{a \in \mathcal{W}}\frac{1}{\sqrt{\mathrm{deg}(t(a))}}\braket{t(a)|g} \notag \\
    &=\frac{1}{\sqrt{|\mathcal{W}|}}\frac{1}{\sqrt{\mathrm{deg}(t(\mathcal{W}))}}\sum_{a \in \mathcal{W}}\braket{t(a)|g} \notag \\
    &=\frac{1}{\sqrt{|\mathcal{W}|\mathrm{deg}(t(\mathcal{W}))}}\frac{|\mathcal{W}|}{\sqrt{|X_{t(\mathcal{W})}|}}\braket{X_{t(\mathcal{W})}|g} \notag \\
    &=\sqrt{\frac{|\mathcal{W}|}{\mathrm{deg}(t(\mathcal{W}))|X_{t(\mathcal{W})}|}} g(t(\mathcal{W})) \notag \\
    &=\sqrt{p(\mathcal{W}^{-1})}g(t(\mathcal{W}))
  \end{align}
  Likewise, the second term of Eq.\ (\ref{eq: elem_W}) is written as 
  \begin{align}
    \label{eq: second_term}
    \bra{\mathcal{W}}SK^{*}\ket{g}
    =\sqrt{p(\mathcal{W})}g(o(\mathcal{W})).
  \end{align}
  By Eqs.\ (\ref{eq: elem_W}), (\ref{eq: first_term}) and (\ref{eq: second_term}), we get
  \begin{align*}
    \psi_{\pm \theta}(\mathcal{W})
    &=\frac{1}{\sqrt{2}|\sin \theta|}\left(\sqrt{p(\mathcal{W}^{-1})\pi_{t(\mathcal{W})}}g(t(a))-e^{\pm i \theta} \sqrt{p(\mathcal{W})\pi_{o(\mathcal{W})}}g(o(a))\right) \\
    &=\frac{\sqrt{M(\mathcal{W})}}{\sqrt{2}|\sin \theta|}(f(t(\mathcal{W}))-e^{\pm i \theta}f(o(\mathcal{W})))
  \end{align*}
  Thus, we get the desired equation.
\end{proof}
Let us show the spectral analysis of $T$. 
Recall $T$ is described by
\begin{align*}
  T=
  \begin{pmatrix}
    0 & x & 0 & 0 & \dots & 0 \\ 
    p_{1} & r_{1} & q_{1} & 0 & \dots & 0 \\
    0 & p_{2} & r_{2} & q_{2} & \dots & 0 \\
    \vdots & \ddots & \ddots & \ddots & \ddots & \vdots \\
    \vdots & \ddots & \ddots & p_{k-1} & r_{k-1} & q_{k-1} \\
    0 & \dots & \dots & \dots & p_{k} & r_{k} \\
  \end{pmatrix}.
\end{align*}
We set $\epsilon=1/d$. Note that $\epsilon$ is sufficiently small when $n$ is large.
By using $\epsilon$, $p_{j}, q_{j}$ and $r_{j}$ are rewritten as 
\begin{align*}
  p_{j}=\epsilon \cdot j^{2}, \quad 
  q_{j}=1-\frac{j}{k}+\epsilon \cdot j(j-k), \quad
  r_{j}=\frac{j}{k}+\epsilon \cdot j(k-2j),
\end{align*}
respectively.
Likewise, $x$ is expanded as
\begin{align*}
  x=\frac{d}{d+m}=1-\epsilon \cdot m+ \epsilon^{2} \cdot m^{2} - \cdots 
\end{align*}
for $m \ll d$.
Thus, $T$ can be expressed as
\begin{align}
  \label{eq: expand of T}
  T=T(\epsilon)=T^{(0)}+\epsilon T^{(1)}+ \epsilon^{2} T^{(2)}+ \cdots.
\end{align}
\begin{align}
  (T^{(0)})_{j, \ell}&=
  \begin{cases*}
    j/k &: $j=\ell$, \\
    1-j/k &: $j=\ell+1$, \\
    0 &: otherwise, 
  \end{cases*}
  \notag
  \\
  \label{eq: T1}
  (T^{(1)})_{j, \ell}&=
  \begin{cases*}
    -m &: $(j,\ell)=(0,1)$ \\
    j(k-2j) &: $j=\ell$, \\
    j(j-k) &: $j=\ell+1$ and $(j,\ell)\ne(0,1)$, \\
    j^{2} &: $j=\ell-1$, \\
    0 &: otherwise.
  \end{cases*}
  \\
  \label{eq: T23}
  (T^{(q)})_{j, \ell}&=
  \begin{cases*}
    (-m)^{q} &: $(j, \ell)=(0,1)$, \\
    0 &: otherwise.
  \end{cases*}
\end{align}
for $j,\ell=0,1,...,k$ and $q=2,3,4,...$.
Note that $T^{(0)}$ is called non-perturbed matrix.
Then, the eigenvalues and eigenvectors of $T^{(0)}$ are given by simple calculation.
\begin{lemma}
  \label{lem: e.val of T}
  The set of the eigenvalues of $T^{(0)}$ is 
  \begin{align*}
    \mathrm{Spec}(T^{(0)})=\left\{\frac{j}{k} \, \middle| \, j=0,1,...,k \right\}.
  \end{align*}
  The right eigenvector associated to $j/k$, $\ket{u^{(j)}}$, is given by
  \begin{equation*}
    u^{(j)}(r)=
    \begin{cases*}
      \left. \dbinom{j}{r} \middle/ \dbinom{k}{r} \right. &: $r=0,1,...,j$, \\
      0 &: $r=j+1,...,k$,
    \end{cases*}
  \end{equation*}
  The left eigenvector associated to $j/k$, $\bra{v^{(j)}}$, is given by
  \begin{equation*}
    v^{(j)}(r)=
    \begin{cases*}
      (-1)^{(r-j)}\dbinom{k-j}{r-j} &: $r=j,j+1,...,k$, \\
      0 &: $r=0,1,...,j-1$. \\
    \end{cases*}
  \end{equation*}
\end{lemma}
We should remark that the right eigenvector $\ket{u}$  and left eigenvector $\bra{v}$ corresponding to the eigenvalue 1 are written as
\begin{align*}
  \ket{u}=[1,1,...,1]^{\top}, \\
  \bra{v}=[0,0,...,0,1],
\end{align*}
respectively.
We show the following proposition, which is the main tool in this proof. 
\begin{prop}
  \label{prop: perturbation}
  {\rm (\cite{Ka82})}

  When $T(\epsilon)$ is expanded as Eq.\ (\ref{eq: expand of T}), the eigenvalue $\lambda(\epsilon)$ of $T(\epsilon)$ induced by the eigenvalue 1 of $T^{(0)}$ is given by
  \begin{align*}
    \lambda(\epsilon)=1+\epsilon \lambda^{(1)}+\epsilon \lambda^{(2)}+ \cdots,
  \end{align*}
  where
  \begin{align}
    \label{eq: lambda}
    &\lambda^{(n)}=\sum_{p=1}^{n}\frac{(-1)^{p}}{p} \sum_{\substack{\nu_{1}+\cdots + \nu_{p}=n \\ \omega_{1}+\cdots + \omega_{p}=p-1}} \mathrm{Tr}\left(T^{(\nu_{1})}S^{(\omega_{1})}\cdots T^{(\nu_{p})}S^{(\omega_{p})}\right), \\
    \label{eq: Sn}
    &S^{(0)}=-P=-\ket{u}\bra{v},\quad  S^{(n)}=S^{n}, \quad S^{(-n)}=O \quad (n\geq 1), \\
    \label{eq: def_S}
    &S=-\sum_{j=0}^{k-1}\frac{1}{1-j/k}\dbinom{k}{j} \ket{u^{(j)}}\bra{v^{(j)}}.
  \end{align}
  Moreover, the right eigenvector $\ket{u(\epsilon)}$ associated to $\lambda(\epsilon)$ is given by
  \begin{align}
    \label{eq: perturbed_vec}
    \ket{u(\epsilon)}=\ket{u}+O(\epsilon).
  \end{align} 
\end{prop}
Note that $\omega_{1},...,\omega_{p}$ are the integers while $\nu_{1},...,\nu_{p}$ are the positive integers including 0.
Proposition \ref{prop: perturbation} implies that the eigenvalue of $T(\epsilon)$ is obtained by using $T$ and $S$ for sufficiently small $\epsilon$.
Additionally, the corresponding eigenvector $\ket{u(\epsilon)}$ can be approximated by non-perturbed eigenvector $\ket{u}$.
\begin{lemma}
  The eigenvalue $\lambda(\epsilon)$ of $T(\epsilon)$ induced by the eigenvalue 1 of $T^{(0)}$ is obtained by
  \begin{align}
    \label{eq: lambda expand}
    \lambda(\epsilon)=1-m \cdot k!\cdot k^{k}\cdot \epsilon^{k+1}+O(\epsilon^{k+2}).
  \end{align}
\end{lemma}
\begin{proof}
  We will consider the following formula
  \begin{align}
    \label{eq: trace}
    \mathrm{Tr}\left(T^{(\nu_{1})}S^{(\omega_{1})}\cdots T^{(\nu_{p})}S^{(\omega_{p})}\right).
  \end{align}
  If there exist $q \in [p]$ such that $\omega_{q}<0$, Eq.\ (\ref{eq: trace}) equals 0 since $S^{(\omega_{q})}=O$.
  Thus, we consider only the case which $\omega_{q} \geq 0$ holds.
  Because of $\omega_{1}+\cdots + \omega_{p}=p-1$, there exist $q \in [p]$ such that $\omega_{q}=0$.
  Without loss of generality, we assume $\omega_{p}=0$.
  Under this assumption, we have
  \begin{align*}
    \mathrm{Tr}\left(T^{(\nu_{1})}S^{(\omega_{1})}\cdots T^{(\nu_{p})}S^{(0)}\right)
    &=\mathrm{Tr}\left(T^{(\nu_{1})}S^{(\omega_{1})}\cdots T^{(\nu_{p})}(-\ket{u}\bra{v})\right) \\
    &=-\bra{v}T^{(\nu_{1})}S^{(\omega_{1})}\cdots T^{(\nu_{p})}\ket{u}.
  \end{align*}
Combining Eqs.\ (\ref{eq: T1}) with (\ref{eq: T23}), we get
\begin{align*}
  T^{(\nu_{p})}\ket{u}=(-m)^{\nu_{p}}\cdot [1,0,...,0]^{\top},
\end{align*}
Hence, it follows
\begin{align*}
  -\bra{v}T^{(\nu_{1})}S^{(\omega_{1})}\cdots T^{(\nu_{p})}\ket{u}
  &=-(-m)^{\nu_{p}}[0,0,...,0,1]T^{(\nu_{1})}S^{(\omega_{1})}\cdots T^{(\nu_{p-1})}S^{(\omega_{p-1})} [1,0,...,0]^{\top}.
\end{align*}
This equation implies that Eq.\ (\ref{eq: trace}) is zero if the $(k,0)$-component of $T^{(\nu_{1})}S^{(\omega_{1})}\cdots T^{(\nu_{p-1})}S^{(\omega_{p-1})}$ is zero.
Since $S,P,T^{(0)}$ and $T^{(q)}\, (q=2,3,...)$ are upper triangular matrices and $T^{(1)}$ is tridiagonal matrix, it must be multiplied by $T^{(1)}$ at least $k$ times in order to have the value without 0 in the $(k,0)$-component of $T^{(\nu_{1})}S^{(\omega_{1})}\cdots T^{(\nu_{p-1})}S^{(\omega_{p-1})}$. 
Thus, Eq.\ (\ref{eq: trace}) becomes non-zero for the first time when $p=k+1$ holds.
On the contrary, it follows that $\mathrm{Tr}\left(T^{(\nu_{1})}S^{(\omega_{1})}\cdots T^{(\nu_{p})}S^{(\omega_{p})}\right)=0$ if $p<k+1$ holds.
Hence, from Eq.\ (\ref{eq: lambda}) and above discussion, we have
\begin{align*}
  \lambda^{(1)}=\cdots=\lambda^{(k)}=0.
\end{align*}
While $\lambda^{(k+1)}$ is as follows:
\begin{align}
\label{eq :lambda-k}
  \lambda^{(k+1)}
  &=\sum_{p=1}^{k+1}\frac{(-1)^{p}}{p} \sum_{\substack{\nu_{1}+\cdots + \nu_{p}=k+1 \\ \omega_{1}+\cdots + \omega_{p}=p-1}} \mathrm{Tr}\left(T^{(\nu_{1})}S^{(\omega_{1})}\cdots T^{(\nu_{p})}S^{(\omega_{p})}\right) \notag \\
  &=\frac{(-1)^{k+1}}{k+1} \sum_{\substack{\omega_{1}+\cdots + \omega_{k+1}=k}} \mathrm{Tr}\left(T^{(1)}S^{(\omega_{1})}\cdots S^{(\omega_{k})} T^{(1)}S^{(\omega_{k+1})}\right) \notag \\
  &=\frac{(-1)^{k+1}}{k+1} (k+1)\sum_{\substack{\omega_{1}+\cdots + \omega_{k}=k}} \mathrm{Tr}\left(T^{(1)}S^{(\omega_{1})}\cdots S^{(\omega_{k})} T^{(1)}S^{(0)}\right) \notag \\
  &=(-1)^{k+1} m \cdot  [0,0,...,0,1] T^{(1)}S^{(\omega_{1})}\cdots T^{(1)}S^{(\omega_{k})} [1,0,...,0]^{\top} \notag \\
  &=(-1)^{k+1} m \left((T^{(1)}S)^{k}\right)_{k,0}
\end{align}
where $\left( A \right)_{i,j}$ is the $(i,j)$-component of a matrix $A$.
Since $S$ is upper triangular matrix and $T^{(1)}$ is tridiagonal matrix, the $(k,0)$- component of $(T^{(1)}S)^{k}$ is given by
\begin{align*}
  \left((T^{(1)}S)^{k}\right)_{k,0}=\prod_{j=1}^{k}(T^{(1)})_{j,j-1} \cdot \prod_{j=0}^{k-1} (S)_{j,j}.
\end{align*}
From Eq.\ (\ref{eq: def_S}), the diagonal component of $S$ is given by
\begin{align*}
  (S)_{j,j}=-\frac{1}{1-j/k} \quad (j=0,1,...,k-1).
\end{align*}
Therefore, we see
\begin{align*}
  \left((T^{(1)}S)^{k}\right)_{k,0}
  &=\prod_{j=1}^{k} j^{2} \cdot \prod_{j=0}^{k-1} \left(-\frac{1}{1-j/k}\right) \\
  &=(k!)^{2} \cdot (-1)^{k} \cdot \frac{k^{k}}{k!} \\
  &=(-1)^{k} \cdot k! \cdot k^{k}.
\end{align*}
Combining Eq.\ (\ref{eq :lambda-k}) with above equation,  $\lambda^{(k+1)}$ is calculated as
\begin{align*}
  \lambda^{(k+1)}=-m\cdot k!\cdot k^{k}.
\end{align*}
Hence, we get Eq.\ (\ref{eq: lambda expand}).
\end{proof}


Let us set $\cos \theta:=\lambda(\epsilon)$.
By Lemma \ref{lem: element} and Eq.\ (\ref{eq: perturbed_vec}), we have the eigenvector associated to $\cos \theta$.
\begin{lemma}
  \label{lem: U_element}
  Let $\ket{\psi_{\pm \theta}} \in \mathrm{Ker}(e^{\pm i \theta}-U)$. Then, it is obtained that
  \begin{align*}
    \psi_{\pm \theta}(\mathcal{W})=\frac{\sqrt{M(\mathcal{W})}}{\sqrt{2}|\sin \theta|}(1-e^{\pm i \theta})
  \end{align*}
  for $\mathcal{W} \in A^{\mathrm{inv}}$.
\end{lemma}
Combining Lemma \ref{lem: U_element} with Eq.\ (\ref{eq: initial state}), we get the following Lemma.
\begin{lemma}
\label{lem: overlap}
  For the initial state $\ket{\psi_{0}}$ and the eigenvector $\ket{\psi_{\pm \theta}}$ of $U$, it holds the following approximated formula:
  \begin{align*}
    \braket{\psi_{0}|\psi_{\pm \theta}}\approx\frac{(1-e^{\pm i \theta})}{\sqrt{2}\sin \theta}\sqrt{\frac{|A_{J}|}{m+|A_{J}|}}.
  \end{align*}
\end{lemma}
\begin{proof}
  By Eq.\ (\ref{eq: initial state}), we have the following equation.
  \begin{equation*}
    \psi_{0}(\mathcal{W})=
    \begin{dcases*}
      \sqrt{\frac{a_{j}|X_{j}|}{|A_{J}|}} &: $\mathcal{W}=\mathcal{A}_{j} \quad (j=0,..., k), $ \\
      \sqrt{\frac{b_{j}|X_{j}|}{|A_{J}|}} &: $\mathcal{W}=\mathcal{B}_{j} \quad (j=0,..., k),$ \\
      \sqrt{\frac{c_{j}|X_{j}|}{|A_{J}|}} &: $\mathcal{W}=\mathcal{C}_{j} \quad (j=0,..., k),$ \\
      0 &: $\mathcal{W}=S_{\pm}$.
    \end{dcases*}
  \end{equation*}
  From the initial state and Lemma \ref{lem: U_element}, we see
  \begin{align*}
    \braket{\psi_{0}| \psi_{\pm \theta}}
    &=\sum_{\mathcal{W} \in A^{\mathrm{inv}}}\psi_{0}(\mathcal{W})\psi_{\pm \theta}(\mathcal{W}) \\
    &=\frac{(1-e^{\pm i \theta})}{\sqrt{2}\sin \theta}\left\{\sum_{j=0} ^{k} \sqrt{M(\mathcal{A}_{j})\cdot \frac{a_{j}|X_{j}|}{|A_{J}|}}\right. \\
    & \qquad + \left.\sum_{j=0}^{k}\sqrt{M(\mathcal{B}_{j})\cdot\frac{b_{j}|X_{j}|}{|A_{J}|}} + \sum_{j=0}^{k}\sqrt{M(\mathcal{C}_{j})\cdot\frac{c_{j}|X_{j}|}{|A_{J}|}}  \right\} \\
    &=\frac{(1-e^{\pm i \theta})}{\sqrt{2}\sin \theta}\left\{\sum_{j=0} ^{k} \sqrt{\frac{a_{j}|X_{j}|}{m+|A_{J}|}\cdot \frac{a_{j}|X_{j}|}{|A_{J}|}}\right. \\
    & \qquad + \left.\sum_{j=0}^{k}\sqrt{\frac{b_{j}|X_{j}|}{m+|A_{J}|}\cdot\frac{b_{j}|X_{j}|}{|A_{J}|}} + \sum_{j=0}^{k}\sqrt{\frac{c_{j}|X_{j}|}{m+|A_{J}|}\cdot\frac{c_{j}|X_{j}|}{|A_{J}|}}  \right\} \\
    &=\frac{(1-e^{\pm i \theta})}{\sqrt{2}\sin \theta}\frac{1}{\sqrt{|A_{J}|(m+|A_{J}|)}} \left\{ \sum_{j=0}^{k} a_{j}|X_{j}| + \sum_{j=0}^{k} b_{j} |X_{j}| + \sum_{j=0}^{k} c_{j}|X_{j}|\right\} \\
    &=\frac{(1-e^{\pm i \theta})}{\sqrt{2}\sin \theta}\frac{1}{\sqrt{|A_{J}|(m+|A_{J}|)}} \left\{d\sum_{j=0}^{k} |X_{j}|  \right\} \\
    &=\frac{(1-e^{\pm i \theta})}{\sqrt{2}\sin \theta}\sqrt{\frac{|A_{J}|}{m+|A_{J}|}}.
  \end{align*}
\end{proof}
Here, by Lemma \ref{lem: overlap}, we get
\begin{align*}
  |\braket{\psi_{0}| \psi_{\pm \theta}}|
  &\approx \left|\frac{e^{\pm i\theta/2}(\mp 2\sin\theta/2)}{\sqrt{2}\cdot 2\sin \theta/2 \cos \theta/2} \right| \\
  &\approx \frac{1}{\sqrt{2}}.
\end{align*}
Therefore, we can estimate the state after $t$ steps as follows:
\begin{align}
  \label{eq: approx state}
  \psi_{t}(\mathcal{W})&=U^{t}\psi_{0}(\mathcal{W}) \notag \\
  &= \sum_{\mu \in \mathrm{Spec}(U)} \mu^{t} \braket{\psi_{0}| \psi_{\mu}} \psi_{\mu}(\mathcal{W}) \notag \\
  &\approx e^{it \theta} \braket{\psi_{0}| \psi_{\theta}} \psi_{\theta}(\mathcal{W})+e^{-it \theta} \braket{\psi_{0}| \psi_{-\theta}} \psi_{-\theta}(\mathcal{W})
\end{align}
for $\mathcal{W} \in A^{\mathrm{inv}}$ and $\psi_{\mu} \in \mathrm{Ker}(\mu-U)$.
By the definition, the probability of finding the star graph after $t$ steps is 
\begin{align*}
    p_{s}(t)=|\psi_{t}(S_{+})|^{2} + |\psi_{t}(S_{-})|^{2}.
\end{align*}
This equation implies that we can get the probability if we estimate $\psi_{t}(S_{\pm})$.
From Lemma \ref{lem: overlap} and Eq.\ (\ref{eq: approx state}), we get
\begin{align}
\label{eq: state star}
    \psi_{t}(S_{-})
    &\approx e^{it \theta} \braket{\psi_{0}|\psi_{\theta}} \braket{S_{-}|\psi_{\theta}}+e^{-it \theta} \braket{\psi_{0}|\psi_{-\theta}} \braket{S_{-}|\psi_{-\theta}} \notag \\
    &=\frac{1}{\sin^{2} \theta}\frac{\sqrt{m|A_{J}|}}{m+|A_{J}|}\left(\cos(t \theta)- \cos((t+1)\theta)\right) \notag \\
    &=\frac{1}{4\sin^{2}\theta/2 \cos \theta/2} \frac{\sqrt{m|A_{J}|}}{m+|A_{J}|} \cdot 2\sin(\theta/2) \sin(t+1/2)\theta \notag \\
    &\approx \frac{1}{\theta}\frac{\sqrt{m|A_{J}|}}{m+|A_{J}|}\sin (t\theta).
\end{align}
Here, from Eq.\ (\ref{eq: lambda expand}), we have
\begin{align}
  \label{eq: approx theta}
  \theta \approx \sin \theta \approx \sqrt{2m \cdot k^{k}\cdot k!\cdot  \epsilon^{k+1}}+O(\epsilon^{k+1}).
\end{align}
On the other hand, since $N=\tbinom{n}{k}\approx(k!)^{-1}n^{k}$ holds for large $n$,
we get
\begin{align}
  \label{eq: N approx}
  N
  &\approx \frac{1}{k!}\left(\frac{1}{k \epsilon} \right)^{k}.
\end{align}
Recall that $d=k(n-k)=\epsilon^{-1}$ holds.
Combining Eq.\ (\ref{eq: state star}), (\ref{eq: approx theta}) and Eq.\ (\ref{eq: N approx}), we have
\begin{align*}
  \psi_{t}(S_{-})
  &\approx \frac{1}{\theta}\frac{\sqrt{m|A_{J}|}}{m+|A_{J}|}\sin (t\theta) \notag \\
  &\approx \frac{1}{\sqrt{2m \cdot k^{k}\cdot k!\cdot  \epsilon^{k+1}}} \frac{\sqrt{m\frac{1}{\epsilon}\frac{1}{k!}\left(\frac{1}{k \epsilon} \right)^{k} }}{m+\frac{1}{\epsilon}\frac{1}{k!}\left(\frac{1}{k \epsilon} \right)^{k}} \sin (t \theta) \notag \\
  &\approx \frac{1}{\sqrt{2m \cdot k^{k}\cdot k!\cdot  \epsilon^{k+1}}} \sqrt{m\cdot k^{k} \cdot k! \epsilon^{k+1}} \sin (t\theta) \notag \\
  &= \frac{1}{\sqrt{2}} \sin (t \theta).
\end{align*}
We should remark that 
\begin{align*}
  \psi_{t+1}(S_{-})=-\psi_{t}(S_{+}),
\end{align*}
holds by the definition of the time evolution matrix.
Thus, we get the probability of finding the star graph after $t$ steps
\begin{align*}
  p_{s}(t)
  &=|\psi_{t}(S_{+})|^{2} + |\psi_{t}(S_{-})|^{2} \\
  &\approx 2 |\psi_{t}(S_{-})|^{2} \\
  &\approx \sin^{2}(t\theta).
\end{align*}
Therefore, the proof of Theorem \ref{thm: prob} is finished.
By Eq.\ (\ref{eq: N approx}), it holds that
\begin{align*}
  \epsilon \approx k^{-1}\cdot \left(k! \cdot N\right)^{-1/k}.
\end{align*}
Thus, the optimal time steps $\tau$ is described as
\begin{align*}
  \tau&=\left \lfloor \frac{\pi}{2\theta} \right \rfloor \\
&\approx \frac{\pi}{2 \sqrt{2m \cdot k^{k}\cdot k!\cdot  \epsilon^{k+1}}} \\
&\approx \frac{\pi}{2 \sqrt{2m \cdot k^{k} \cdot k! \cdot (k^{-1} (k!\cdot N)^{-1/k})^{k+1}}} \\
&= \frac{\pi\sqrt{k\cdot (k!)^{1/k}}}{2\sqrt{2m}}\sqrt{N^{1+1/k}}, \\
\end{align*}
Hence, the proof of Theorem \ref{thm: time} is finished.

\qed

\section{Summary and Discussion}
\label{sec: summary}
In this work, we proposed the pulsation of discrete-time quantum walks on graphs, which is a generalization of the quantum search algorithms.
On the composite graph formed by $J(n,k)$ and $S_{m}$, we proved that the pulsation exists with the asymptotic probability $1+o(1)$ (Theorem \ref{thm: prob}).
In addition, we find that the order of the optimal time with respect to the number of vertices $N$ is $O(\sqrt{N^{1+1/k}})$ (Theorem \ref{thm: time}).

The pulsation is observed on several graphs by numerical simulation.
Figure 2 shows examples of the pulsation in the case of $Q_{n} \wedge S_{m}$ and $K_{n_{1}} \wedge K_{n_{2}}$,
where $Q_{n}$ is the $n$-dimensional hypercube.
We consider that such a phenomenon may give a solution to reveal the mathematical structure of the success of the search. 
Thus, to clarify the behavior of the pulsation on the general graph is one of the future problems. 
The details of the pulsation on the composite graph formed by more than three graphs is another interesting future's works.

\begin{figure}[htbp]
  \begin{minipage}[b]{0.5\linewidth}
    \centering
    \includegraphics[keepaspectratio, scale=0.55]{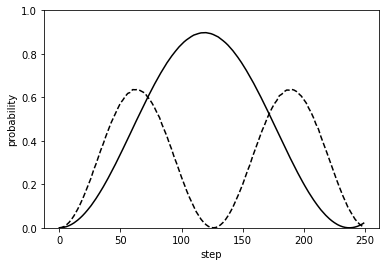}
    \subcaption{}
  \end{minipage}
  \begin{minipage}[b]{0.5\linewidth}
    \centering
    \includegraphics[keepaspectratio, scale=0.55]{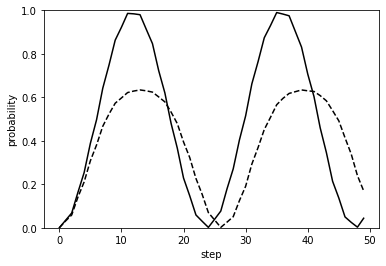}
    \subcaption{}
  \end{minipage} 
  \caption{
  These are examples on the composite graphs formed by (a) $Q_{10}$ and $S_{m}$, (b) $K_{30}$ and $K_{n}$, respectively.
  (a): The solid and dot curves is the probability of the finding $S_{m}$ corresponding to the case of $m=1$ and $m=5$, respectively.
  (b): The solid and dot curves is the probability of the finding $K_{n}$ corresponding to the case of $n=30$ and $n=60$, respectively.
  }
\end{figure}


\section*{Acknowledgments}
 E.S. acknowledges financial supports from the Grant-in-Aid of Scientific Research (C) Japan Society for the Promotion of Science (Grant No. 24K06863).
The authors have no competing interests to declare that are relevant to the content of this article. All data generated or analyzed during this study are included in this published article.


\begin{thebibliography}{1}

  \bibitem{ADZ93}
  Y.~Aharonov, L.~Davidovich, and N.~Zagury.
  \newblock Quantum random walks.
  \newblock Phys. Rev. A, 48(2):1687--1690, 1993.
  
  \bibitem{AAKV01}
  D.~Aharonov, A.~Ambainis, J.~Kempe, and U.~Vazirani.
  \newblock Quantum walks on graphs.
  \newblock In Proc. 33th STOC, pages 50--59, New York, 2001. ACM.
  
  \bibitem{AKR05}
  A.~Ambainis, J.~Kempe, and A.~Rivosh.
  \newblock Coins make quantum walks faster.
  \newblock In Proc. 16th Annual ACM-SIAM Symposium on Discrete Algorithms
    SODA, pages 1099--1108, 2005.
  
  \bibitem{A07}
  A.~Ambainis.
  \newblock Quantum walk algorithm for element distinctness.
  \newblock SIAM J. Comput., 37(1):210-239, 2007.
  
  \bibitem{AGFK20}
  A. Ambainis, A. Gilyen, S. Jeffery, and M. Kokainis. 
  \newblock Quadratic speedup for finding marked vertices by quantum walks.
  \newblock In Proceedings of the 52nd Annual ACM SIGACT Symposium on Theory of Computing, STOC 2020, pages 412–424, New York, 2020.
  
  \bibitem{BCN89}
  A. E. Brouwer, A. M. Cohen, and A. Neumaier.
  \newblock {\it Distance-Regular Graphs.}
  \newblock Modern Surveys in Mathematics. Springer, Berlin, 1989.
  
  \bibitem{BH12}
  A. E. Brouwer and W. H. Haemers.
  \newblock {\it Spectra of Graphs.}
  \newblock Springer, New York, 2012.
  
  \bibitem{CG04}
  A.M. Childs and J. Goldstone.
  \newblock Spatial search by quantum walk.
  \newblock Phys. Rev. A, 70:022314, 2004.

  \bibitem{CGGV15}
  G. Coutinho, C. Godsil, K. Guo and F. Vanhove.
  \newblock Perfect state transfer on distance-regular graphs and association schemes.
  \newblock Linear Algebra Appl., 478:108–130, 2015.

  \bibitem{GR01}
  C. Godsil and G.F. Royle.
  \newblock {\it Algebraic Graph Theory.}
  \newblock Springer-Verlag, New York, 2001.
  
  \bibitem{G12}
  C. Godsil.
  \newblock State transfer on graphs.
  \newblock Discrete Math., 312(1), pages 129-147, 2012.

 \bibitem{GGKL20}
  C. Godsil, K. Guo, M. Kempton, G. Lippner and F. M¨unch. 
  \newblock State transfer in strongly regular graphs with an edge perturbation.
  \newblock J. Comb. Theory, Series A, 172:105181, 2020.

  \bibitem{HKSS14}
  Yu.~Higuchi, N.~Konno, I.~Sato and E.~Segawa.
  \newblock Spectral and asymptotic properties of {G}rover walks on crystal lattices.
  \newblock J. Funct. Anal., 267(11):4197 -- 4235, 2014.

  \bibitem{HKSS17}
  Yu. Higuchi, N. Konno, I. Sato and E. Segawa.
  \newblock Periodicity of the discrete-time quantum walk on a finite graph.
  \newblock Interdiscip. Inf. Sci., 23, pages 75-86, 2017.

  \bibitem{HPS24}
  T. Hosaka, R. Portugal and E. Segawa.
  \newblock Sensitivity of quantum walk to phase reversal and geometric perturbations: an exploration in complete graphs.
  \newblock Interdiscip. Inf. Sci., in press, 2024, arXiv:2402.08243.

  \bibitem{IK05}
  N. Inui, N. Konno. 
  \newblock Localization of multi-state quantum walk in one dimension.
  \newblock Phys. A: Stat. Mech. Appl. 353, pages 133-144, 2005.

  \bibitem{Ka82}
  T. Kato.
  \newblock {\it A Short Introduction to Perturbation Theory for Linear Operator.}
  \newblock Springer-Verlag, New York, 1982.

  \bibitem{Ki22}
  C. Kiumi.
  \newblock Localization of space-inhomogeneous three-state quantum walks.
  \newblock J. Phys. A: Math 55(22) ,2022.

  \bibitem{K02}
  N. Konno.
  \newblock Quantum random walk in one dimension.
  \newblock Quantum Inf. Process., 1, pages 345-354, 2002.
  
  \bibitem{K05}
  N. Konno.
  \newblock A new type of limit theorems for the one-dimensional quantum random walk.
  \newblock J. Math. Soc. Japan, 57(4), pages 1179-1195, 2005.
  
  \bibitem{KLS13}
  N. Konno, T. Luczak and E. Segawa.
  \newblock Limit measures of inhomogeneous discrete-time quantum walks in one dimension.
  \newblock Quantum Inf. Process., 12, pages 33-53, 2013.
  
  \bibitem{KSTY18}
  S. Kubota, E. Segawa, T. Taniguchi, Y. Yoshie.
  \newblock Periodicity of Grover walks on generalized Bethe trees.
  \newblock Linear Algebra Its Appl., 554, pages 371-391, 2018.

  \bibitem{KS22}
  S. Kubota and E. Segawa.
  \newblock Perfect state transfer in Grover walks between states associated to vertices of a graph.
  \newblock Linear Algebla Its Appl., 646, pages 238-251, 2022.

  \bibitem{KY24}
  S. Kubota and K. Yoshino.
  \newblock Symmetry of graphs and perfect state transfer in Grover walks.
  \newblock arXiv:2402.17341, 2024.
  
  \bibitem{M15}
  T. Machida. 
  \newblock Limit theorems of a 3-state quantum walk and its application for discrete uniform measures.
  \newblock Quantum Inf. Comput., 15, pages 406-418, 2015.
  
  \bibitem{R18}
  R.~Portugal.
  \newblock {\it Quantum Walks and Search Algorithms.}
  \newblock Springer, Cham, 2nd edition, 2018.
  
  \bibitem{SKW03}
  N.~Shenvi, J.~Kempe, and K.~B. Whaley.
  \newblock A quantum random walk search algorithm.
  \newblock Phys. Rev. A, 67(5):052307, 2003.

  \bibitem{SS23}
  M. \v{S}tefa\v{n}\'{a}k and S. Skoup\'{y}.
  \newblock Quantum walk state transfer on a hypercube.
  \newblock Phys. Scr., 98, 104003, 2023.

  \bibitem{S04}
  M.~Szegedy.
  \newblock Quantum speed-up of Markov chain based algorithms.
  \newblock In Proc. 45th Annual IEEE Symposium on Foundations of Computer Science, FOCS ’04, pages 32--41, Washington, 2004.
  
  
  \bibitem{TSP22D}
  H.~Tanaka, M.~Sabri, and R.~Portugal.
  \newblock Spatial search on {J}ohnson graphs by discrete-time quantum walk.
  \newblock J. Phys. A: Math. Theor., 55(25):255304, 2022.
  
  \bibitem{XRL19}
  X. Xue, Y. Ruan, and Z. Liu.
  \newblock Discrete-time quantum walk search on Johnson graphs.
  \newblock Quantum Inf. Process., 18(2):50, 2019.

  \bibitem{Y19}
  Y. Yoshie.
  \newblock Periodicity of Grover walks on distance-regular graphs.
  \newblock  Graphs Combin., 35, pages 1305-1321, 2019.

  
  \end{thebibliography}
\end{document}